\documentclass[envcountsame]{llncs}

\usepackage{anysize,graphicx}
\usepackage[table]{xcolor}
\usepackage[utf8]{inputenc}
\usepackage{amsmath}
\usepackage{amssymb}
\usepackage{algorithmic}
\usepackage{algorithm}
\usepackage{bbm}

\usepackage{framed}
\usepackage{eurosym, amsfonts}
\definecolor{shadecolor}{gray}{0.9}
\definecolor{lightgraytable}{gray}{0.9}

\RequirePackage{tikz}
\usetikzlibrary{arrows,positioning,automata}
\usepackage{wrapfig}
\usepackage{url}
\usepackage{bbold}
\usepackage{array}
\newcolumntype{L}[1]{>{\raggedright\let\newline\\\arraybackslash\hspace{0pt}}m{#1}}
\newcolumntype{C}[1]{>{\centering\let\newline\\\arraybackslash\hspace{0pt}}m{#1}}
\newcolumntype{R}[1]{>{\raggedleft\let\newline\\\arraybackslash\hspace{0pt}}m{#1}}

\renewcommand{\S}{{\mathcal S}}
\newcommand{\A}{{\mathcal A}}
\newcommand{\R}{{\mathbb R}}

\newcommand{\setN}{{\mathbb N}}

\newcounter{Theorem}
\newtheorem{assumption}[Theorem]{Assumption}

\title{The Stochastic Shortest Path Problem: \\ A polyhedral combinatorics perspective }

\author{ Matthieu Guillot \and  Gautier Stauffer}
\authorrunning{Guillot and Stauffer}

\institute{Univ. Grenoble Alpes, G-SCOP,
    38000 Grenoble,
    France.
    \email{\{matthieu.guillot,gautier.stauffer\}@g-scop.grenoble-inp.fr} 
}

\begin{document}
\date{}
%\vspace{-20ex}
\maketitle

\begin{abstract}
In this paper, we give a new framework for the stochastic shortest path problem in finite state and action spaces. Our framework generalizes both the frameworks proposed by Bertsekas and Tsitsiklis \cite{BertsekasTsitsiklis} and by Bertsekas and Yu \cite{Bertsekas16}. We prove that the problem is well-defined and (weakly) polynomial when (i) there is a way to reach the target state from any initial state and (ii) there is no transition cycle of negative costs (a generalization of negative cost cycles). These assumptions generalize the standard assumptions for the deterministic shortest path problem and our framework encapsulates the latter problem (in contrast with prior works). In this new setting, we can show that (a) one can restrict to deterministic and stationary policies, (b) the problem is still (weakly) polynomial through linear programming, (c) Value Iteration and Policy Iteration converge, and (d) we can extend Dijkstra's algorithm.  

\end{abstract}

%\newpage 
\section{Introduction}

The Stochastic Shortest Path problem (SSP) is a Markov Decision Process (MDP) that generalizes the classic deterministic shortest path problem. We want to control an agent, who evolves dynamically in a system composed of different {\em states}, so as to converge to a predefined  {\em target}. The agent is controlled by taking {\em actions} in each time period\footnote{We focus here on discrete time (infinite) horizon problems.}~: actions are associated with costs and transitions in the system are governed by probability distributions that depend exclusively on the previous action taken and are thus independent of the past. We focus on finite state/action spaces~: the goal is to choose an action for each state, a.k.a. a {\em deterministic and stationary policy}, so as to minimize the total expected cost incurred by the agent before reaching the (absorbing) target state, when starting from a given initial state.

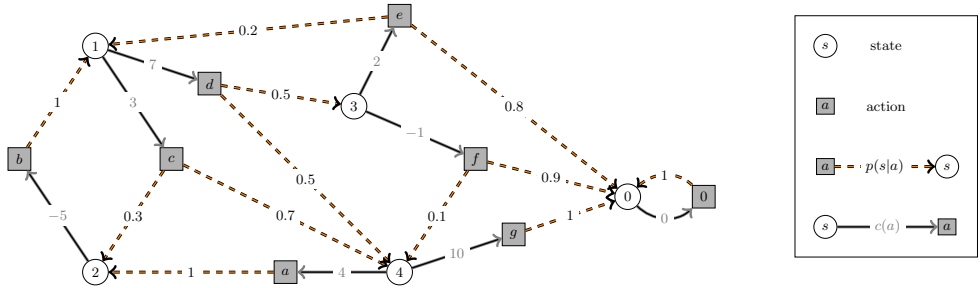
\begin{figure}[h]
\begin{center}
\scalebox{0.8}{
\begin{tikzpicture}
  \begin{scope}[scale=2.5,shape=circle,minimum size=0.5cm,fill=white]
    \tikzstyle{every node}=[draw, scale=0.75] 
    \node (1) at (0,2) {$1$};
    \node (2) at (0,0.5) {$2$};
    \node (4) at (2,0.5) {$4$};
    \node (3) at (1.7,1.6) {$3$};
    %\node (5) at (3.5,1) {$0$} edge [looseness=1,scale=0.5, double=black,in=45,out=-45,loop] node[draw=none,fill=white]  {0} ();
    \node (5) at (3.5,1) {$0$};
    \node (a) at (1.25,0.5)  [rectangle,fill=black!30] {$a$};
    \node (b) at (-0.5,1.25)  [rectangle, fill=black!30]{$b$};
    \node (c) at (0.5,1.25)  [rectangle, fill=black!30]{$c$};
    \node (d) at (0.75,1.75)  [rectangle, fill=black!30]{$d$};
    \node (e) at (2,2.2)  [rectangle, fill=black!30]{$e$};
    \node (f) at (2.5,1.25)  [rectangle, fill=black!30]{$f$};
    \node (g) at (2.75,0.75)  [rectangle, fill=black!30]{$g$};
    \node (h) at (4,1)  [rectangle, fill=black!30]{$0$};
  \end{scope}
  
\tikzset{mystyle/.style={->,dashed, double=orange}} 
\tikzset{mystyle2/.style={->,gray,double=black}} 
\tikzset{every node/.style={fill=white,scale=0.75}} 
\tikzset{mystyle3/.style={->,dashed, double=orange,in=45,out=135}} 
\tikzset{mystyle4/.style={->,gray,double=black,in=-135,out=-45}}

  \path (1)     edge [mystyle2]    node   {$3$} (c)
      (c)     edge [mystyle]    node   {$0.7$} (4) 
      (4)     edge [mystyle2]    node   {$10$} (g)
      (g)     edge [mystyle]    node   {$1$} (5);
 
  \path (c)     edge [mystyle]    node   {$0.3$} (2)
      (2)     edge [mystyle2]    node   {$-5$} (b)
      (b)     edge [mystyle]    node   {$1$} (1)
      (1)     edge [mystyle2]    node   {$7$} (d)
      (d)     edge [mystyle]    node   {$0.5$} (3)
      (3)     edge [mystyle2]    node   {$-1$} (f)
      (f)     edge [mystyle]    node   {$0.9$} (5);

  \path (3)     edge [mystyle2]    node   {$2$} (e)
      (e)     edge [mystyle]    node   {$0.2$} (1);
  \path (4)     edge [mystyle2]    node   {$4$} (a)
      (a)     edge [mystyle]    node   {$1$} (2);

  \path (d)     edge [mystyle]    node   {$0.5$} (4);
   \draw (e) edge[mystyle] node   {$0.8$} (5);
   \draw (f) edge[mystyle] node   {$0.1$} (4);
   
   \path (5) edge [mystyle4]    node   {$0$} (h) 
    (h)     edge [mystyle3]    node   {$1$} (5);
   
   \draw (11.5,1.5) rectangle (14.5,5.5);
   \node (a''') at (12,4) [rectangle, draw,fill=black!30] {$a$};
   \node (s''') at (12,5) [shape=circle,draw,fill=white] {$s$};
   \node (text1) at (13,4) {action};
   \node (text2) at (13,5)  {state};
   %\node (text3) at (13,1)  {Legend};

   \node (a') at (12,3) [rectangle, draw,fill=black!30] {$a$};
   \node (s') at (14,3) [shape=circle,draw,fill=white] {$s$};
   \draw (a')  edge[mystyle] node {$p(s|a)$} (s');

   \node (a'') at (14,2) [rectangle, draw,fill=black!30] {$a$};
   \node (s'') at (12,2) [shape=circle,draw,fill=white] {$s$};
   \draw (s'')  edge[mystyle2] node {$c(a)$} (a'');

%  \draw [->] (1) -- (c);
%  \draw [->] (1) -- (d) ;
%  \draw [->] (3) -- (e) ;
%  \draw [->] (4) -- (g) ;
%  \draw [->] (4) -- (a) ;
%  \draw [->] (2) -- (b) ;
%  \draw [->] (3) -- (f) ;
%  \draw [->] (a) -- (2) ;
%  \draw [->] (b) -- (1) ;
%  \draw [->] (c) -- (2) ;
%  \draw [->] (c) -- (4) ;
%  \draw [->] (d) -- (3) ;
%  \draw [->] (d) -- (4) ;
%  \draw [->] (f) -- (4) ;
%  \draw [->] (f) -- (5) ;
%  \draw [->] (e) -- (1) ;
%  \draw [->] (e) -- (5) ;
%  \draw [->] (g) -- (5) ;

%  \draw (q_A) -- (q_1) -- (q_2) -| (q_E);
  %\draw[->,shorten >=2pt] (q_A) .. controls +(75:1.4cm) and +(105:1.4cm) .. node[above] {$x$} (q_A);
\end{tikzpicture}}
\end{center}
\caption{A graphical representation of a SSP (with target state $0$)~: circles are states, squares are actions, dashed arrows indicate state transitions (probabilities) for a given action, and black edges represent actions available in a given state with corresponding cost.}\label{fig:1}
\end{figure}

More formally, a stochastic shortest path instance is defined by a tuple $(\mathcal{S}, \mathcal{A},J,P,c)$ where $\S=\{0,1,\ldots,n\}$ is a finite set of {\em states},  $\A=\{0,1,\ldots,m\}$ is a finite set of {\em actions}, $J$ is a 0/1 matrix with $m$ lines and $n$ columns and general term $J(a,s)$, for all $a\in \{1,...,m\}$ and $s\in \{1,...,n\}$, with $J(a,s)=1$ if and only if action $a$ is available in state $s$, $P$ is a {\em row substochastic matrix} with $m$ lines and $n$ columns and general term $P(a,s):=p(s|a)$ (probability of ending in $s$ when taking action $a$), for all  $a\in  \{1,...,m\}$, $s\in  \{1,...,n\}$, and a cost vector $c \in \R^{m}$. The state $0$ is called the {\em target} state and the action $0$ is the unique action available in that state. Action $0$ lead to state $0$ with probability $1$. When confusion may arise, we denote state $0$ by $0_\S$ and action $0$ by $0_\A$. A {\em row substochastic} matrix is a matrix with nonnegative entries so that every row adds up to at most $1$. We denote by $\mathcal{M}_{\leq}(l,k)$  the set of all $l \times k$ row substochastic matrices and by $\mathcal{M}_{=}(l,k)$ the set of all row stochastic matrices ($i.e.$ for which every row adds up to exactly $1$). In the following, we denote by $\A(s)$ the set of actions available from $s\in \{1,...,n\}$ and we assume without loss of generality\footnote{If not we simply duplicate the actions.} that for all $a \in \mathcal{A}$, there exists a unique $s$ such that $a \in \mathcal{A}(s)$. We denote by $\A^{-1}(s)$ the set of actions that lead to $s$ i.e. $\A^{-1}(s):=\{a: P(a,s)>0\}$. 

We can associate a directed bipartite graph $G=(\S,\A,E)$ with $(\mathcal{S}, \mathcal{A}, J, P)$ by defining $E:=\{(s,a): s\in \S\setminus\{0\},a\in \A\setminus\{0\} \mbox{ with } J(a,s)=1\} \cup \{(a,s): s\in \S\setminus\{0\},a\in \A^{-1}(s)\} \cup \{(0_\S,0_\A),(0_\A,0_\S)\}$. $G$ is called the {\em support graph}. A {\em $\S$-walk} in $G$ is a sequence of vertices $(s_0,a_0,s_1,a_1,...,s_k)$ for some $k\in \mathbb N$ with $s_i\in \S$ for all $0\leq i\leq k$, $a_i\in \A$ for all $0\leq i\leq k-1$, $(s_i,a_i)\in E$ for all $0\leq i\leq k$, and $(a_{i-1},s_i) \in E$ for all $1\leq i \leq k$. $k$ is called the {\em length} of the walk. We denote by $W_k$ the set of all possible $\S$-walk of length $k$ and $W:=\cup_{k\in \setN} W_k$. A {\em policy} $\Pi$ is a function $\Pi : (k,w_k) \in \setN  \times W_k \mapsto \Pi_{k,w_k}\in \mathcal{M}_{=}(n,m)$ satisfying $\Pi_{k,w_k}(s,a)>0 \implies J(s,a)=1$ for all $s\in \{1,...,n\}$ and $a\in \{1,...,m\}$.  We say that a policy is {\em deterministic} if $\Pi_{k,W_k}$ is a 0/1 matrix for all $k$ and $w_k$, it is  {\em randomized} otherwise. If $\Pi$ is constant, we say that the policy is {\em stationary}, otherwise it is called {\em history-dependent}. A policy $\Pi$ induces a probability distribution over the (countable) set of all possible $\S$-walks. When $\Pi$ is stationary, we often abuse notation and identify $\Pi$ with a matrix. %A $U$-walk $(s_0,a_0,...,s_t)$ is {\em terminal} if $s_t=0$ (by definition $s_0,...,s_{t-1} \neq 0 as there are no action available in state $0$). 

We let $y^\Pi_k \in \R_+^n$ be the substochastic vector representing the state of the system in period $k$ when following policy $\Pi$ (from an initial distribution $y^\Pi_0$). That is $y^\Pi_k(i)$ is the probability of being in state $i$, for all $i=1,...,n$ at time $k$ following policy $\Pi$. Similarly, we denote by $x_k^{\Pi} \in \mathbb{R}_+^m$ the substochastic vector representing the probability to perform action $j$, for all $j=1,...,m$, at time $k$ following policy $\Pi$. By the law of total probability, we have  $x_k^{\Pi}=\Pi^T \cdot y^\Pi_{k}$ for all $k\geq 0$. Similarly, given a stationary policy $\Pi$ and an initial distribution $y^\Pi_0$ at time $0$, we have ~: $y^\Pi_k=P^T x_{k-1}^{\Pi}=  P^T \cdot \Pi^T \cdot y^\Pi_{k-1}$ for all $k\geq 1$.  Hence the state of the system at time $k\geq 0$ follows $y^\Pi_k= (P^T \cdot \Pi^T)^k \cdot y^\Pi_0$. The value $c^T x_{k}^{\Pi}$ represents the expected cost paid at time $k$ following policy $\Pi$. One can define for each $i\in \S\setminus\{0\}$, $J_\Pi(i):=\limsup_{K\rightarrow +\infty}  \sum_{k=0}^{K} c^T x_{k}^{\Pi}$ with $y^\pi_0:=e_i$, and $J^*(i):=\min \{J_\Pi(i) : \Pi $ deterministic  and stationary policy$\}$\footnote{$\limsup$ is used here as the limit need not be defined in general.} ($e_{i}$ is the characteristic vector of $\{i\}$ i.e. the 0/1 vector with $e_i(j)=1$ iff $j=i$). Bertsekas and Tsistiklis \cite{BertsekasTsitsiklis} defined a stationary policy $\Pi^*$ to be {\em optimal}\footnote{note that it is not clear, a priori, that such a policy exists} if $J^*(i):=J_{\Pi^*}(i)$ for all $i\in \S\setminus \{0\}$. They introduced the {\em Stochastic Shortest Path Problem} as the problem of finding such an optimal stationary policy. 

\subsection*{Literature review}

The stochastic shortest path problem is a special case of Markov Decision Process  and it is also known as total reward undiscounted MDP \cite{bertsekasbook1,bertsekasbook2,puterman}.  It arises naturally in robot motion planning, from maneuvering a vehicle over unfamiliar terrain, steering a flexible needle through human tissue or guiding a swimming micro-robot through turbulent water for instance \cite{Robotmotion}. It has also many applications in operations research, artificial intelligence and economics: from inventory control, reinforcement learning to asset pricing (see for instance \cite{white93,Merton,bauerle2011markov,SuttonBarto}). SSP forms an important class of MDPs as it contains finite horizon MDPs, discounted MDPs (a euro tomorrow is worth less than a euro today) and average cost problems (through the so-called vanishing discounted factor approach) as special cases. It thus encapsulates most of the work on finite state/action MDPs. The stochastic shortest path problem was introduced first by Eaton and Zadeh in 1962 \cite{EatonZadeh} in the context of pursuit-evasion games and it was later studied thoroughly by Bertsekas and Tsitsiklis \cite{BertsekasTsitsiklis}.

%MDPs are powerful tools to model and solve problems under uncertainty \cite{bertsekasbook1, bertsekasbook2,puterman}. Those models have been extensively used in particular in operations research, robotics and control, artificial intelligence and economics: from inventory control, robot motion  planning, reinforcement learning to asset pricing (see for instance \cite{white93,Merton, bauerle2011markov,SuttonBarto}). SSP forms an important class of MDPs as it contains finite horizon MDPs, discounted MDPs (a euro tomorrow is worth less than a euro today) and average cost problems (through the so-called diminushing discounted factor approach) as special cases. It thus encapsulates most of the work on finite state/action MDPs. The problem is also known as first passage problem and was introduced initially in the context of pursuit-evasion problems \cite{EatonZadeh}. 

MDPs were first introduced in the 50's by Bellman \cite{Bellman57} and Shapley \cite{Sha53} and they have a long, rich and successful history (see for instance \cite{bertsekasbook1,bertsekasbook2,puterman}). For most MDPs, it is known that there exists an optimal deterministic and stationary policy \cite{puterman}. Building upon this fact, there are essentially three ways of solving such problems exactly (and some variants)~: {\em value iteration} (VI), {\em policy iteration} (PI) and linear programming (LP). Value iteration and policy iteration are the original 50+ years old methods \cite{Bellman57,Howard}.  The idea behind VI is to approximate the infinite horizon problem with a longer and longer finite one. The solution to the $k$-period
%\begin{wrapfigure}{l}{0.15\textwidth}
%%\begin{figure}[h]
% % \begin{center}
%    \includegraphics[trim=0cm 0cm 0cm 0cm,clip,width=0.15\textwidth]{polytope_path}
%    %\caption{A polyhedra and a path to an optimal solution}
% % \end{center}
%  %\caption{polytope}
%\end{wrapfigure}
%\end{figure}
approximation is built inductively from the optimal solution to the $(k-1)$-period problem using standard dynamic programing. The convergence of the method relies mainly on the theory of contraction mappings and Banach fixed-point theorem \cite{bertsekasbook1} for most MDPs.
PI is an alternative method that starts from a feasible deterministic and stationary policy and iteratively improves the action in each state so as to converge to an optimal solution. It can be interpreted as a simplex algorithm where multiple pivots are performed in each step \cite{Manne,Denardo}. As such it builds implicitly upon the geometry of the problem to find optimal solutions. Building explicitly upon this polyhedra, most MDPs can also be formulated as linear programs and as such they can thus be solved in (weakly) polynomial time \cite{Manne,depenoux,Denardo,Hordijk,Lasserre}.

In the context of the SSP, some hypothesis are required for standard methods and proof techniques to apply. Bertsekas and Tsitsiklis \cite{BertsekasTsitsiklis} introduced the notion of {\em proper} stationary policies~: a stationary policy $\Pi$ is said to be {\em proper} if ${\bf 1}^T (P^T \cdot \Pi^T)^n \cdot e_i < 1$ for all $i=1,...,n$, that is, after $n$ periods of time, the probability of reaching the target state is positive, from any initial state $i$. We say that such policies are {\em BT-proper} as we will introduce a slight generalization later. They proved that VI, PI and LP  still work when two assumptions hold, namely, when (i) there exists a BT-proper policy and (ii) any BT-improper policy $\Pi$ have at least one state $i$ for which $J_\Pi(i)=+\infty$. In particular they show that one can restrict to deterministic policies. Their assumptions naturally discriminate between BT-proper and BT-improper policies. Exploiting further the discrepency between these policies, Bertsekas and Yu \cite{Bertsekas16} showed that one can relax asumptions (i) and (ii) when the goal is to find an optimal {\em BT-proper} stationary policy. They could show that applying the standard VI and PI methods onto a perturbated problem where $c$ is modified to $c+\delta\cdot {\bf 1}$ with $\delta>0$ and letting $\delta$ tends to zero over the iterations, yields an optimal BT-proper solution if (j) there exists a BT-proper policy and $(jj)$ $J^*$ is real-valued. Moreover they could also show that the problem can still be formulated (and thus solved) using linear programming, which settles the (weak) polynomiality of this extension. Some authors from the AI community proposed alternative extensions of the standard SSP introduced by Bertsekas and Tsitsiklis. It is easy to see that the most general one, entitled Stochastic and Safety Shortest Path problem \cite{S3P}, is a special case of Bertsekas and Yu's framework (it is a bi-objective problem that can be easily model in this framework using artificial actions of prohibited cost).

The question of whether SSP, in its original form or the later generalization by Bertsekas and Yu, can be solved in strongly polynomial time\footnote{a polynomial in the number of states and the number of actions} is a major open problem for MDPs (see for instance \cite{Ye11}). It was proven in a series of breakthrough papers that it is the case for {\it fixed} discount rate (basically the same problem as before but where the transition matrix $P$ is such that there is a fixed non-zero probability of ending up in $0$ after taking any action). The result was first proved using interior point methods \cite{Ye05}  and then the same author showed that the original policy iteration method proposed by Howard was actually strongly polynomial too \cite{Ye11} (the analysis was later improved \cite{Hansen13}). The problem is still open for the undiscounted case but Policy Iteration is known to be exponential in that setting \cite{Friedmann}. In contrast, value iteration was proved to be exponential even for the discounted case \cite{FeinbergJeff14}. Because SSPs can be formulated as  linear programs, the question relates very much to the existence of strongly polynomial time algorithms for linear programming, a very long-lasting open problem that was listed as one of the 18 mathematical problems of the 21st century by Smale in 1998 \cite{Smale98}. A possible line of attack is to study simplex-type of algorithms but existence of such algorithms is also a long standing open problem and relates to the Hirsch conjecture on the diameter of polyhedra. These questions are central in optimization, discrete geometry and computational complexity. Despite the fact that SSP exhibits strong additional properties over general LPs, these questions are still currently out of reach in this setting, too.

In practice, value iteration and policy iteration are the methods of choice when solving medium size MDPs.
For large scale problems (i.e. most practical applications), approximate solutions are needed to provide satisfying solutions in a reasonable amount of time \cite{Powell}.  The field is known as Approximate Dynamic Programming and is a very active area of research. Most approximation methods are based on approximate versions of exact algorithms and developping new exact approaches is thus of great practical interest.

In this paper, we propose an extension of the frameworks of Bertsekas and Tsitsiklis \cite{BertsekasTsitsiklis} and Bertsekas and Yu \cite{Bertsekas16}. We prove in section \ref{sec:existence} that, in this setting, there is an optimal deterministic and stationary policy. Then we show in section \ref{sec:algo} that the standard Value Iteration and Policy Iteration methods converge, and we give an alternative approach that generalizes Dijkstra's algorithm when the costs are non negative.

\subsection{Notations and definitions}

Given a directed graph $G(V,E)$, and a set $S\subset V$, we denote by $\delta^+(S)$ the set of arcs $(u,v)$ with $u\in S$ and $v\not\in S$, and by $N^+(S)$ the set of vertices $v\in V$ such that $(u,v)\in E$ for some $u\in S$. For convenience when $S$ is a singleton, we denote $\delta^+(\{u\})$ by $\delta^+(u)$ and $N^+(\{u\})$ by $N^+(u)$. Then we define inductively $N_k^+(u):=N^+(N_{k-1}^+(u))\setminus N_{k-1}^+(u) \cup \ldots \cup N_0^+(u))$ for $k\geq 1$ integer with $N_0^+(u)=\{u\}$. We denote by $R^+(u)$ the set of vertices {\em reachable} from $u$ i.e. $R^+(u)=\bigcup_{k\geq 0} N_k^+(u)$. We can define $\delta^-(u)$, $N^-(u)$, $N_k^-(u)$ and $R^-(u)$ analogously. Clearly $v\in R^-(u)$ if and only if $u\in R^+(v)$. $R^-(u)$ are the vertices that can reach $u$. When confusion may arise, we denote $R^+(u)$ by $R_G^+(u)$ (and similarly for the other notations). We denote by $\mathbbm{1}_A$ the indicator function associated with a set $A$ i.e. $\mathbbm{1}_A$ is a 0/1 function with $\mathbbm{1}_A(x)=1$ if and only if $x\in A$. For a vector $x \in \mathbb{R}^d$ and $I\subseteq \{1,...,d\}$ we denote by $x[I]$ the restriction of $x$ to the indices in $I$ and $x(I):=\sum_{i\in I} x(i)$.

\section{Our new framework}\label{sec:ext}

We start with a simple observation whose proof can be found in the Appendix (see section \ref{sec:A1}).

\begin{lemma}\label{lem:well-def} \mbox{For BT-proper stationary policies, $\displaystyle \lim_{K\rightarrow +\infty} \sum_{k=0}^{K} x_{k}^{\Pi}$ is finite for any initial state distribution $y_0^\Pi$.}
\end{lemma}

%\begin{proof}
%$(I-Q) \cdot \sum_{k=0}^K Q^k = I - Q^{K+1}$. Now because $\lim_{k\rightarrow +\infty} Q^k=0$, we have for $K$ large enough, $det((I-Q) \cdot \sum_{k=0}^K Q^k)=det(I-Q)\cdot det(\sum_{k=0}^K Q^k)=det(I - Q^{K+1})>0$. It follows that $det(I-Q)\neq 0$ and the result follows.
%\end{proof}

We now extend the notion of proper policies introduced by Bertsekas and Tsitsiklis using this alternative (relaxed) property and {\em from now on we will only use this new definition}. 

Given a state $s\in \{1,...,n\}$, a policy $\Pi$ is said to be {\em $s$-proper} if $\sum_{k\geq 0} x_{k}^{\Pi}$ is finite, when $y^\Pi_0:= e_{s}$. Observe that $\sum_{k\geq 0} y_{k}^{\Pi}$ is also finite for s-proper policies (as $y_{k}^{\Pi}=P^T x_{k-1}^{\Pi}$). In particular $\lim_{k\rightarrow +\infty} y^\Pi_{k}=0$, that is the policy lead to the target state $0$ with probability $1$ from state $s$.  The {\em $s$-stochastic-shortest-path problem }($s$-SSP for short) is the problem of finding a $s$-proper policy $\Pi$ of minimal cost $c^T \sum_{k\geq 0} x_{k}^{\Pi}$. We say that a policy is {\em proper} if it is $s$-proper for all $s$ and {\em improper} otherwise. The {\em stochastic shortest path problem} (SSP) is the problem of finding a proper policy $\Pi$ of minimal cost $c^T \sum_{k\geq 0} x_{k}^{\Pi} $ where $y^\Pi_0:=\frac{1}{n} \bf 1$. It is easily seen that the {stochastic shortest path problem}, as defined here, is also a special case of the $s$-SSP as one can add an artificial state with only one action that leads to all states in $\{1,...,n\}$ with probability $\frac{1}{n}$. In the following two sections, unless otherwise stated, we restrict to the $s$-SSP. In this context, we often abuse notation and we simply call proper a $s$-proper policy.

%\subsection{A natural linear programming relaxation}

Since for any policy $\Pi$ (possibly history-dependent and randomized), $\Pi_{k,w_k}$ are stochastic matrices, we have at any period $k\geq 0$,  $\sum_{a \in \A(s)} x_k^{\Pi}(a) = y_k^{\Pi}(s)$. We also have $y_{k+1}^{\Pi}(s)=\sum_{a\in \A} p(s|a) x_k^{\Pi}(a)$ for all $s\in \{1,...,n\}$. In matrix form this is equivalent to $y_{k}^{\Pi} =  J^T x_k^{\Pi} $ and $y_{k+1}^{\Pi}= P^T x_k^{\Pi}$. This implies $J^T x_{k+1}^{\Pi} = P^T x_k^{\Pi}$ for all $k\geq 0$. We also have $J^T x_0^\Pi = e_s$. Now $x^\Pi:=\sum_{k=0}^{\infty}{x_{k}^{\Pi}}$ is well-defined for proper policies. Summing up the previous relations over all periods $k \geq 0$ we get $(J-P)^T {x^{\Pi}} = e_{s}$.   Hence the following linear program is a relaxation of the $s$-SSP problem\footnote{We would like to stress on the fact that the LP relaxation we consider here is almost (except for the right hand side) the standard LP formulation of the problem of finding an optimal deterministic and stationary policy and it was already known for quite some time for many special cases of SSP (see \cite{Bertsekas16} for instance). However while in the MDP community, the LP formulation comes as a corollary of other results, here we reverse the approach and introduce this formulation as a natural relaxation of the problem and we derive the standard results as (reasonably) simple corollaries. This is what allows to simplify, generalize and unify many results from the litterature. This is a simple yet major contribution of this paper. The notation and terminology is taken from \cite{Hansen}}
.

\begin{equation}\tag{$P_{s}$}
\begin{array}{ll}
\min & c^Tx \\
(J-P)^T x & = e_{s} \\
x & \geq 0
\end{array}
\end{equation}

Observe that for a deterministic problem (i.e. when $P$ is a 0/1 matrix), $(J-P)^T$ is the node-arc incidence matrix of a graph (up to a row) and the corresponding LP is the standard network flow relaxation of the deterministic shortest path problem. The vector $x$ is sometimes called a network {\em flux} as it generalizes the notion of network flow.

%As pointed out in Puterman \cite{puterman}, a fundamental question in Markov Decision Problem theory is: ``For a given optimality criterion, under what conditions is it optimal to restrict to deterministic and stationary policies ?''. We will show that we can restrict to such policies for the stochastic shortest path problem when there is no ``transition cycles'' of negative cost. Our approach is original in the sense that it uses a standard approach from polyhedral combinatorics~: we take a linear relaxation of the problem and we show that the extreme points correspond to deterministic and stationary policies. This proves that when the relaxation has finite optimal value, then there exists an optimal, deterministic and stationary policy. 

We call  a solution $x$ to $(J-P)^T x = 0,x\geq 0$ a {\em transition cycle} and the cost of such a transition cycle $x$ is $c^Tx$. Negative cost transition cycles are the natural extension of negative cost cycles for deterministic problems (actually we could also consider only the extreme rays of $(J-P)^T x = 0,x\geq 0$ ; we defer the discussion to the journal version of the paper).  One can check the existence of such objects by solving a linear program.

\begin{lemma}
One can check in (weakly) polynomial time whether a stochastic shortest path instance admits a negative cost transition cycle through linear programming.
\end{lemma}

We will prove in the sequel that the extreme points of $P_s:=\{x\geq 0 : (J-P)^T x  = e_{s}\}$ `correspond' to proper deterministic and stationary policies. Hence, when the relaxation $(P_s)$ has a finite optimum (i.e. when there is no transition cycle of negative cost and when a proper policy exists), this will allow to prove that, the $s$-SSP admits an optimal proper policy which is deterministic and stationary. This answers, for this problem, one fundamental question in Markov Decision Problem theory ``Under what conditions is it optimal to restrict to deterministic and stationary policies ?''  \cite{puterman}.

%Observe that for a proper policy $\sum_{k\geq 0} x_{k}^{\Pi} \ c$ is well-defined. % since $u_K:=\sum_{K\geq k\geq 0} x_{k}^{\Pi}$ is non decreasing and bounded from above.

We can assume without loss of generality that there exists a path between all state node $i$ and $0$ in the support graph $G$. Indeed, if there is a node $i$ with no path to $0$ in  $G$, then no $s$-proper policy will pass through $i$ at any point in time (because then the probability of reaching the target state, starting from $i$, is zero, contradicting $\lim_{k\rightarrow +\infty} y^\Pi_{k}= 0$)  ; we could thus remove $i$ and the actions leading to $i$ and iterate. It is easy to see that under this assumption, there is always a $s$-proper policy. Indeed the randomized and stationary policy $\Pi$ that chooses an action in state $i$ uniformly at random among $\A(i)$ will work~: in this case, for each state $i$, there is in fact a non zero probability of choosing one of the paths from $i$ to $0$ after at most $n$ periods of time. 

\begin{lemma}\label{lem:enough}
Consider a $s$-SSP instance where there exists a path between all state node $i$ and $0$ in the support graph $G$. Then the policy that consists, for each state $i\in \S\setminus\{0\}$, in choosing uniformly at random an action in $\A(i)$ is a proper stationary policy.
\end{lemma}

The discussion above also gives a simple algorithm for testing the existence of a proper policy for any instance of the SSP.

\begin{lemma}
One can check in time $O(|U|\cdot (|U|+|V|+|E|))$ whether a $s$-SSP instance with support graph $G=(U,V,E)$ admits a proper policy or not.
\end{lemma}

We are now ready to introduce the new assumptions that we will use to study the stochastic shortest path problem. They are the very natural extensions of the standard assumptions for the deterministic shortest path problem.\\

\begin{assumption}\label{as:1}
\begin{shaded}
We consider $s$-SSP/SSP instances where :
\begin{itemize}
\item there exists a path between all state node $i$ and $0$ in the support graph $G$, and
\item there is no negative cost transition cycle.
\end{itemize}
\end{shaded}
\end{assumption}

As already observed, these assumptions can be checked in (weakly) polynomial time. Moreover, these assumptions  implies that $(P_s)$ has  a finite optimum (from standard LP arguments). Also Bertsekas and Yu's framework is a special case of our setting as in the presence of negative cost transition cycles, $J^*(i)$ is not real-valued for some state $i$\footnote{In order to prove this statement formally we shall prove that if $J^*(i)$ is real-valued for all $i$ then there is no negative cost transition cycle~: we can prove that when there exists a negative cost transition cycle, we can find one which is `induced' by  a (non proper) deterministic and stationary policy $\Pi$  and that all vertices $i$ on this cycle will have $J_\Pi(i)=-\infty$ ; for this, we need to extend our decomposition result to decompose transition cycles into extreme rays ; we leave the details for the journal version of the paper, but  the proof of Proposition \ref{lem:proper2} gives the flavor of this latter result.}. The main extension, with respect to Bertsekas and Yu, is that we allow for non-stationary proper policies in the first place.

\section{Existence of an optimal, deterministic and stationary policy}\label{sec:existence}

In this section, we will prove essential properties about $P_s:= \{x \geq 0: (J-P)^T x = e_{s}\}$. This will allow to prove that, under Assumption \ref{as:1}, we can restrict to optimal proper, deterministic and stationary policies. The following theorem can be seen as an extension of the {\em flow decomposition theorem} (see  \cite{networkflow}).

%$P_{s}$ is a polyhedron in $\R^{m}$ and as such it can be decomposed by Minkowski-Weyl and Minkowski-Carath\'eodory Theorems (see for instance  \cite{Schrijver86}) into the sum of the convex hull of at most $|A|+1$ vectors and a vector in the {\em recession cone}, i.e. a solution to $(J-P)^T x  = 0, x\geq 0$. We will design a strongly polynomial-time algorithm to find such a decomposition for any solution $x$ in $P_{s}$~: \\

\begin{theorem}\label{th:decomp}
Let $x \in \mathbb{R}^m$ be a feasible solution of $(P_{s})$. In strongly polynomial time, one can find $1\leq k \leq m$,  $x_1, ..., x_k, x_c \in \mathbb{R}^m$, and $\lambda_1, ... , \lambda_k \in [0, 1]$ such that $x_1, ..., x_k$ are feasible solutions of $(P_{s})$, $x_c$ satisfies $(J-P)^Tx_c = 0, x_c\geq 0$, $\sum_{j = 1}^k{\lambda_j} = 1$ and $x = \sum_{j=0}^k{\lambda_j x_j} + x_c$. 
Moreover, the vectors $x_j$ are network flux corresponding to proper, deterministic and stationary policies, i.e. for all $j\in 1,\dots, k$, there exists a proper, deterministic and stationary policy $\Pi_j$ such that $x_j = x^{\Pi_j}$.
\end{theorem}

Before we can prove this theorem, we need a couple of useful lemmas and definitions. Let $G=(U,V,E)$ and $x \in \mathbb{R}^m$ be a  solution to $(J-P)^T x = e_s, x\geq 0$. Let $G_x$ be the subgraph of $G$ induced by the vertices in $\S \cup \A_x$ where $\A_x:=\{a\in \{1,...,m\}$ with $x(a) >0$. $G_x$ is called the {\em support graph of $x$ in $G$}. We denote by $E_x$ the set of edges of $G_x$. 

\begin{lemma}\label{lem:path}
There exists a path between all states reachable from $s$ in $G_x$ and $0_\S$. In other word, for all $i \in R_{G_x}^+(s)$, we have $i \in R_{G_x}^{-}(0_\S)$.
\end{lemma}\begin{proof}
Let us define $\bar{x}\in \R^{|E_x|}$ as follows~: $\bar{x}((s',a)):=x(a)$ for all $a\in A_x$ and $s'$ the (unique) state with  $a\in \A(s')$, and $\bar{x}((a,s')):=P(a,s')\cdot x(a)$ for all $a\in A_x$, and $s'\in S$ such that $P(a,s)>0$. Observe that $\bar{x}$ is only defined on $E_x$  and that $\bar{x}>0$. Because $x$ is a feasible solution to $(P_s)$, $\bar{x}$ satisfies $\bar{x}(\delta_{G_x}^+(v)) - \bar{x}(\delta_{G_x}^-(v))= \mathbbm{1}_{\{s\}}(v) -  \mathbbm{1}_{\{0_\S\}}(v)$ for all $v\in G_x$ and $\bar{x} \geq 0$. It is thus a  unit $(s,0_\S)$-flow in $G_x$. Now let us assume that there exists $i\in R_{G_x}^+(s)$ with  $i \not\in R_{G_x}^{-}(0)$. Summing up all flow constraints over $v\in R^+(i)$, we get $\bar{x}(\delta^+(R^+(i))) - \bar{x}(\delta^-(R^+(i))) = \mathbbm{1}_{R^+(i)}(s)$ (we remove from now on the subscript $G_x$ in order not to overload the notation).  We have $\bar{x}(\delta^+(R^+(i)))=0$ by definition of $R^+(i)$. But then $\bar{x}(\delta^-(R^+(i))) +\mathbbm{1}_{R^+(i)}(s)= 0$. Since $\bar{x}(\delta^-(R^+(i)))\geq 0$, this implies $s\not\in R^+(i)$ and  $\bar{x}(\delta^-(R^+(i)))=0$. Now because $s\not\in R^+(i)$ and $s\in R^-(i)$ (by hypothesis), there is at least one arc of $E_x$ in $\delta^-(R^+(i))$ but this implies $\bar{x}(\delta^-(R^+(i)))>0$ as  $\bar{x} >0$, a contradiction.
\end{proof}

Given a proper, deterministic and stationary policy $\Pi$, we denote by $G_{\Pi}$ the subgraph of $G$ induced by the state vertices in $\S$ and the actions vertices in $\Pi$. Now let  $G^s_{\Pi}$ be the subgraph of $G_\Pi$ induced by the vertices in $R^+(s)$. $G^s_\Pi$ is called the {\em support graph} of $\Pi$ (it is easily seen that it corresponds to the subgraph induced by the states and actions that we might visit under policy $\Pi$). Because $\Pi$ is proper, $0_S$ is reachable from each state $i$ in $G^s_{\Pi}$. Let us denote by $\S'$ the state vertices in $G^s_{\Pi}$ and $\Pi(\S')$ the actions associated with $\S'$ in $\Pi$. We also denote by $P_{\S'}$ the restriction of $P$ to the rows in $\S'$ and the columns in $\Pi(\S')$ ($P_{\S'}$ is a $|\S'|\times |\S'|$ matrix). Following the same arguments as in Section \ref{sec:A1}, $\lim_{k\rightarrow +\infty} (P_{\S'})^k = 0 $ and thus $(I_{\S'}-P_{\S'})$ is invertible. Now observe that $(I_{\S'}-P_{\S'})^T x^\Pi[\Pi(S')]=e'_s$ for $x^\Pi:=\sum_{k=0}^{+\infty} x^\Pi_k$, with $y_0^\Pi:=e_s$ ($e'_s$ is the restriction of $e_s$ to the indices in $\S'$). Indeed $x^\Pi(a)=0$ for all $a\not\in\Pi(S')$ and thus $(I_{\S'}-P_{\S'})^T x^\Pi[\Pi(S')]=e'_s$ corresponds to the constraints of $(P_s)$ associated with the rows in $\S'$. We thus have the following lemma.%The flux vector $x^\Pi$ thus satisfies $x^\Pi[\Pi(S')]= (I_{\S'}-P_{\S'})^{-T}  e'_s$ and $x^\Pi(a)=0$ for all $a\not\in \Pi(\S')$.

\begin{lemma}\label{lem:flux}
Given a proper, deterministic and stationary policy $\Pi$, the flux vector $x^\Pi$ associated with $\Pi$ and defined by $x^\Pi:=\sum_{k=0}^{+\infty} x^\Pi_k$, with $y_0^\Pi:=e_s$ satisfies $x^\Pi[\Pi(S')]= (I_{\S'}-P_{\S'})^{-T}  e'_s$ and $x^\Pi(a)=0$ for all $a\not\in \Pi(\S')$, with $\S'$,$\Pi(\S')$, $I_{\S'}$,$P_{\S'}$ and $e'_s$  defined as above.
\end{lemma}

The following Lemma is  easy to prove using similar flow arguments as in the proof of Lemma \ref{lem:path}.

\begin{lemma}\label{lem:sub}
Let $\Pi$ be a proper, deterministic and stationary policy. We have $G^s_\Pi=G_{x^\Pi}$. Moreover if $x\in P_{s}$ and $\Pi(S) \subseteq \A_x$, then $G^s_{\Pi}$ is a subgraph of $G_x$ and $x^\Pi(a) \geq x(a)$ for some $a\in \A_x$.
\end{lemma}

%\begin{corollary}\label{cor:degeneration}
%Let $\Pi$ be a s-proper, deterministic and stationary policy. Then the flux vector $x^{\Pi}$ is an extreme point of the constraint polyhedra of $(P_s)$.
%\end{corollary}
%
%\begin{remark}\label{rmk:degeneration}
%There can be degeneration in the constraint polyhedra of $(P_s)$. Indeed, for a state $i\notin S'$, any action $j \in \mathcal{A}(i)$ can be chosen by $\Pi$ without changing the flux vector $x^{\Pi}$.
%\end{remark}

Before proving Theorem \ref{th:decomp}, we need a final Lemma.

\begin{lemma}\label{lem:proper}
Let $G=(U,V,E)$ be the support graph of a $s$-SSP instance and assume that there is a path from every state vertex $i$ to $0_\S$ in $G$. Then in time $O(|U|+|V|+|E|)$, one can find a proper, deterministic and stationary policy $\Pi$.
\end{lemma}
\begin{proof}
We know that, $0\in R^+(i)$ for all $i$, is enough to ensure that there is a proper policy by Lemma \ref{lem:enough}. Now if there exists a state vertex $i$ in $G$ with $|\A(i)|>1$, we can  delete from $G$ an action in $\A(i)$ that does not remove $0$ from $R^+(i)$. Such an action exists as it is enough to keep an action $a\in A(i)$ with minimum distance to $0$ (in terms of arc) to ensure that $0$ is still in $R^+(i)$ after deletion (by  minimality of the distance to $0$, such an action has a directed path to $0$ that does not go through $i$). If $|\A(i)|=1$ for all $i$ then the only possible policy is proper (from Lemma \ref{lem:enough}), deterministic and stationary. We can implement such a procedure in time $O(|U|+|V|+|E|)$ by computing $N^-_k(0)$ for all $k\leq |U|+|V|$ and a  $0$-anti-arborescence $A$ using a breadth first search algorithm~: we then keep only the actions in $A$. \end{proof}

We are now ready to prove the main Theorem of this section.

\begin{theorem}\label{th:decomp}
Let $x \in \mathbb{R}^m$ be a feasible solution of $(P_{s})$. In strongly polynomial time, one can find $1\leq k \leq m$,  $x_1, ..., x_k, x_c \in \mathbb{R}^m$, and $\lambda_1, ... , \lambda_k \in [0, 1]$ such that $x_1, ..., x_k$ are feasible solutions of $(P_{s})$, $x_c$ satisfies $(J-P)^Tx_c = 0, x_c\geq 0$, $\sum_{j = 1}^k{\lambda_j} = 1$ and $x = \sum_{j=0}^k{\lambda_j x_j} + x_c$. 
Moreover, the vectors $x_j$ are network flux corresponding to proper, deterministic and stationary policies, i.e. for all $j\in 1,\dots, k$, there exists a proper, deterministic and stationary policy $\Pi_j$ such that $x_j = x^{\Pi_j}$.
\end{theorem}
\begin{proof}
We prove first that such a decomposition exists for any $x\in P_s$. Let $x$ be a smallest counter-example (in terms of $|A_x|$). Because $x$ is a feasible solution of $(P_{s})$, we know by Lemma \ref{lem:path} that there exists a path between all states reachable from $s$ in $G_x$ and $0$. Now from Lemma \ref{lem:proper}, we know that there exists a proper, deterministic and stationary policy $\Pi$ to which we can associate and compute a flux $x^\Pi$ using Lemma \ref{lem:flux}. Let $\lambda\geq 0$ be the maximum value such that $x':=x-\lambda x^\Pi\geq 0$.  By Lemma \ref{lem:sub} we have that $G_{x^\Pi}$ is a subgraph of $G_x$ and thus $\lambda>0$ (as $x>0$ on $\A_x$). We also have $\lambda\leq 1$ by the same Lemma. Moreover by maximality of $\lambda$, there is an arc $a\in \A_x$ such that $x(a)>0$ and $x'(a)=0$. Hence $\A_{x'}\subset \A_x$. If $\lambda=1$, $x'$ is a solution to $(J-P)^T x = 0,x\geq 0$ and $x:=x^\Pi+x'$ provides a decomposition for $x$, a contradiction. Else, $\frac{1}{1-\lambda} x'$ is a solution to $(P_s)$ with $|\A_{x'}|<|\A_{x}|$. By minimality of the counter-example, we can assume that there exists a decomposition for $\frac{1}{1-\lambda} x'$. Now we can get a decomposition for $x$ from the decomposition for $\frac{1}{1-\lambda} x'$ by scaling the multipliers by 1-$\lambda$ and using $x^\Pi$ with multiplier $\lambda$, this is contradiction. Clearly, we can make the proof algorithmic and because $A_{x'}\subset A_x$ at each iteration, the algorithm will terminate in at most $|A_x|$ steps with a set of $k\leq |A_x|$ solutions $x_1,....,x_k$ to $(P_{s})$ and a vector $x_c$ satisfying the theorem. 
\end{proof}

\begin{corollary}\label{cor:detstat}
Under Assumption \ref{as:1}, the $s$-SSP admits an optimal proper, deterministic and stationary policy.
\end{corollary}
\begin{proof}
We know from linear programming that when a LP has finite optimum, we can find an optimal solution in an extreme point. For $(P_s)$ this is guaranteed by Assumption \ref{as:1}. But an extreme point $x$ of $P_s$ cannot be expressed as a convex combination of other points of $P_s$ by definition. As such, using Theorem \ref{th:decomp}, $x$ must be equal to $x^{\Pi}$ for some proper, deterministic and stationary policy $\Pi$. Now $c^T x^{\Pi}$ is precisely the cost of policy $\Pi$. Hence we have a feasible solution to our original problem which is optimal for the linear relaxation $(P_s)$. It is thus optimal for the original problem.
\end{proof}

We can deduce from what preceeds a result which is standard for the deterministic shortest path problem~: {\em Bellman optimality conditions}. 

\begin{lemma}\label{lem:Bellman}
Let $\Pi$ be an optimal proper, deterministic and stationary solution to the $s$-SSP (under Assumption \ref{as:1}). Let $G^s_\Pi$ be the support graph of $\Pi$. For all state vertex $i$ in $G^s_\Pi$, $\Pi$ is optimal for $i$-SSP.
\end{lemma}
\begin{proof}
Observe first that $i$-SSP satisfies Assumption \ref{as:1}. Now suppose $\Pi$ is not optimal for $i$-SSP. We know from Corollary \ref{cor:detstat} that  $i$-SSP admits an optimal proper, deterministic and stationary policy $\Pi_i$. Now the (history-dependent and non stationary) policy $\Pi'$ that consists in applying policy $\Pi$ to problem $s$-SSP, up to when state $i$ is reached (if it ever is) and then applying policy $\Pi_i$ is a proper policy. The value of this policy is better than the value of $\Pi$ as there exists a realization where $i$ is reached, a contradiction.
\end{proof}

%\begin{lemma}\label{lem:dual}
%Let $(D_s)$ be the dual of $(P_s)$ :
%\begin{equation}\tag{$D_{s}$}
%\begin{array}{ll}
%\max & z_s \\
%(J-P) z & \leq c \\
%z \in \mathbb{R}^n
%\end{array}
%\end{equation}
%
%and $z^*$ be a optimal solution of $(D_s)$. Then for all $i \in \mathcal{S}$, $z^*_i$ is the optimal value of state $i$. Moreover, if $x^*$ is an optimal solution of $(P_s)$, then for all $j \in \mathcal{A}$ such that $x^*_j >0$, $z^*(N^-(j)) =  c(j) + \sum_{i \in \mathcal{S}}{p(i|j)z*(i)}$.
%\end{lemma}
%
%\begin{definition}
%Let $j \in \mathcal{A}$. The reduced cost of action $j$ is defined by $\bar{c}_j = c_j - (z_{N^-(j)} - \sum_{i \in N^+(j)}{p(i|j)z_i})$
%\end{definition}
%
%
%With remark \ref{rmk:degeneration}, we saw that there can be degeneration in the constraint polyhedra of $(P_s)$. Do avoid it, we consider now a SSP problem (and not a SS s-P).
%
\section{Algorithms}\label{sec:algo}

We focussed, up to now and without loss of generality, on the $s$-SSP problem. Bellman optimality conditions (i.e. Lemma \ref{lem:Bellman}) also tells us that, under Assumption \ref{as:1}, we can actually restrict attention to the SSP problem as well without loss of generality. Indeed we already observed that $SSP$ can be converted to a $s$-SSP problem by simply adding an artificial state $s$ and a unique action available from $s$ that lead to all states $i=1,...,n$ with probability $\frac{1}{n}$. Now there is a one-to-one correspondance between the policies of SSP and the policies of the auxiliary $s$-SSP problem and hence any proper, deterministic and stationary solution $\Pi$ to $SSP$ is optimal if and only if it is optimal for the auxiliary problem.  But by Lemma \ref{lem:Bellman}, an optimal policy $\Pi^*$ for SSP is optimal for $i$-SSP for all $i=1,...,n$ (as all $i$ are in $G'_{\Pi^*}$). It is easy to see that all theorems from the previous section extend naturally to the SSP setting. Of course, some definitions and results have to be slightly adapted~: for instance, the flux vector $x^\Pi$ associated with a proper deterministic and stationary policy is now $x^\Pi:=\sum_{k=0}^{+\infty} x^\Pi_k$ with $y_0^\Pi:=\frac{1}{n} \bf 1$ and it satisfies $x^\Pi= (I-P_{\Pi})^{-T}  \frac{1}{n} \bf 1$ (see Lemma \ref{lem:flux} for the previous relation), where $P_\Pi$ is the $n\times n$ matrix obtained from $P$ by keeping only the rows corresponding to actions in $\Pi$. For algorithmic reasons, it is more convenient to deal with the SSP problem as there is no problem of degeneracy~: the feasible basic solution $x^\Pi$ (it is indeed now the basic solution associated with the basis $(I-P_\Pi)^T$) has positive values on the actions in $\Pi$.  In this section, we will therefore focus on the SSP problem. The corresponding linear programming formulation is (in principle, the right hand side should be $\frac{1}{n} {\bf  1}$ but we simply rescaled it): 

\begin{equation}\tag{$P$}
\begin{array}{ll}
\min & c^Tx \\
(J-P)^T x & =  {\bf 1} \\
x & \geq 0
\end{array}
\end{equation}

One possible way of solving the previous model is to use any polynomial time algorithm for linear programming. This would lead to weakly polynomial time algorithms for SSP. As pointed out in the introduction, there are two standard alternatives for solving a MDP~: Value Iteration and Policy Iteration. We prove in the next two sections the convergence of these methods under Assumption \ref{as:1}. Then we give another new iterative method based on the standard primal-dual approach to linear programming~: this can be considered as a natural generalization of Dijkstra's algorithm.

\subsection{Value Iteration}

We denote by $\mathcal P$ the set of all  proper  policies for SSP. For all $i=1,...,n$, we define $V^*(i)$ to be the optimal value of $(P_i)$ (again under Assumptions \ref{as:1}), i.e. $V^*(i):={\min_{\Pi \in \mathcal P} c^T x^\Pi }$ with $y_0^\Pi=e_i$. This is refered to as the {\em value} of state $i$. We have in particular $V^*(i)= \min_{\Pi \in \mathcal P} \lim_{K \rightarrow +\infty} {\sum_{k = 0}^{K}{c^T x_k^{\Pi}}}$ by definition of $x^\Pi_k$. In the following,  we show that we can switch the $\min$ and $\lim$ operators with some care. We need first to introduce an auxiliary SSP instance obtained from $(\mathcal{S}, \mathcal{A},J,P,c)$ by adding an action of cost $M(i)$ for each state $i=1,...,n$  that lead to state $0$ with probability one, with $M(i)$ ``big enough'' . We call aux-SSP this auxiliary problem. Observe that in aux-SSP, there are proper policies that terminate in at most $k$ time periods for all $k\geq 1$, from any starting state. Indeed one can always chose an auxiliary action in period $k-1$.  Let us denote by $\mathcal P^k$ the proper policies in aux-SSP that terminate in at most $k$ steps and by $\mathcal P_{aux}$ the proper policies for aux-SSP.  Observe that $V_K(i):=\min_{\Pi \in \mathcal{P}^K} {\sum_{k = 0}^{K}{c^T x_k^{\Pi}}}$ is well-defined for each $K\geq 1$. In fact it is easy to prove by induction that it follows the dynamic programming formula~: $V_k(i)=\min\{V_{k-1}(i),\min_{a\in \A(i)} c(a) + \sum_{j} p(j|a) V_{k-1}(j)\}$ for all $k\geq 2$ and  $V_{1}(i)=M(i)$ for all $i=1,...,n$ (an optimal, deterministic non-stationary policy $\Pi^*_K$ can be recovered easily too)~: $V_k(i)$ is indeed the optimal value starting from $i$ among policies in $\mathcal P^k$.  The following result can be seen as an extension of Bellman-Ford algorithm for the deterministic shortest path problem. We give the proof in Section \ref{sec:A4}.

\begin{theorem}\label{th:VI}
For all $i=1,...,n$, if $M(i)\geq V^*(i)$, then we have $V^*(i)=\displaystyle   \lim_{K \rightarrow +\infty} V_K(i).$
\end{theorem}

Notice that it is easy to find initial values for $M(i)$ satisfying the previous Theorem. Indeed one can use $V^\Pi(i)$, the values for state $i$ when using policy $\Pi$ for any $i$-proper policy $\Pi$. We can actually  easily find a proper deterministic and stationary policy for SSP (i.e. for all $i$ simultaneously) by extending Lemma \ref{lem:proper} to SSP. 

The algorithm that consists in evaluating $V_k$ iteratively until $||V_{k} - V_{k-1}||_\infty$ is below some threshold is called {\em Value Iteration}. Value Iteration was already known to converge for SSP in the presence of transition cycles of cost zero, when initialized appropriately, see Bertsekas and Yu \cite{Bertsekas16}. 

We now explain how to recover an optimal proper, deterministic and stationary policy given the optimal vector $V^*$. Let us consider the dual linear program $(D)$ of $(P)$~:
 
\begin{equation}\tag{$D$}
\begin{array}{ll}
\max & {\bf 1}^T y \\
(J-P) y & \leq c  \\
\end{array}
\end{equation}

By definition of $V^*(i)$ and by Lemma \ref{lem:Bellman}, we know that $V^*$ satisfies $V^*(i)= \min_{a\in \A(i)} c(a) + \sum_{j} p(j|a) V^*(j)$ for all $i=1,...,n$. Also extending Corollary \ref{cor:detstat} to SSP, we know that there exists an optimal proper deterministic and stationary policy $\Pi^*$ with $V^*(i)=  c(\Pi^*(i)) + \sum_{j} p(j|\Pi^*(i)) V^*(j)$ for all $i=1,..,n$. In particular, $y^*:=V^*$ is feasible for $(D)$ and because the pair $(x^{\Pi^*}, y^*)$ satisfies the complementary slackness conditions, $y^*$ is optimal for $(D)$. 

Now let us reverse the complementary slackness conditions. An optimal solution $x^*$ to $(P)$ can have $x^*(a)>0$ only if $V^*(i)= c(a) + \sum_{j} p(j|a) V^*(j)$. Let $\A^*$ be the set of all such actions and let us restrict our instance of SSP to those actions in $\A^*$. Because there is an optimal proper, deterministic and stationary policy $\Pi^*$ for SSP and because such a policy must use only actions in $\A^*$, we know that there is a path from every state to the target state $0$ in the support graph $G^*=(U^*,V^*,E^*)$ of this instance. Now the stationary policy $\Pi'$ consisting in choosing uniformly at random an action in $\A^*(i)$ for each state $i$ is proper by Lemma \ref{lem:enough} and $x^{\Pi'}:=\lim_{K\rightarrow +\infty} \sum_{k=0}^K x_k^{\Pi'}$ is feasible for $(P)$. Observe that by construction, the pair $(x^{\Pi'},y^*)$ satisfies the complementary slackness conditions and thus $x^{\Pi'}$ is also optimal for $(P)$. Extending theorem \ref{th:decomp} to SSP and by optimality of $x^{\Pi'}$, we can decompose $x^{\Pi'}$ into a convex combination of vectors $x^{\Pi_j}$ associated with proper, deterministic and stationary policies $\Pi_j$. Again by optimality of $x^{\Pi'}$, $x^{\Pi_j}$ are also optimal for $(P)$. We can thus use  Theorem \ref{th:decomp} to get an optimal, proper, deterministic and stationary policy for the problem. In fact we can stop after the first application of  Lemma \ref{lem:proper} and thus get such a policy in time $O(|U^*|+|V^*|+|E^*|$).

N.B. We can define an approximate proper solution $\Pi_k$ at each step $k$ of Value Iteration by considering an approximate version of the complementary slackness theorem. We defer the discussion to the journal version of the paper.

\subsection{Policy Iteration}

An alternative to Value Iteration is to use a simplex algorithm to solve $(P)$. In order to do so we need an initial basis. We can use Lemma \ref{lem:proper} to find a proper deterministic and stationary policy $\Pi$. Then as we already observed, $x^\Pi=(I-P_\Pi)^{-T} \bf 1$ is a non-degenerate feasible basic solution of $(P)$. Because the basic solutions are non-degenerate, we can implement any pivot rule from this initial basic solution and the simplex algorithm will converge in a finite number of steps.  This type of algorithm is often referred to as {\em simple policy iteration} in the litterature. This proves that simple PI terminates in a finite number of steps. Unfortunately, most pivot rules are known to be exponential in $n$ and $m$ in the worst case \cite{melekopoglou1994complexity}.

In contrast with simple policy iteration, Howard's original policy iteration method \cite{Howard} changes the actions of a (basic) policy in each state $i$ for which there is an action in $\A(i)$ with negative {\em reduced cost}. We will prove now that this method converges under Assumptions \ref{as:1}. For this, we will prove that the method iterates over proper deterministic and stationary policies and that the cost is decreasing at each iteration. Given a proper deterministic and stationary policy $\Pi$,  $x^\Pi=(I-P_\Pi)^{-T} \bf 1$ is the basic feasible solution of $(P)$ associated with the basis $(I-P_\Pi)^{T}$. We define the {\em reduced cost} vector associated with $c$ and $\Pi$ as ${\bar{c}}^{\Pi}:= c - c_\Pi (I-P_\Pi)^{-T} (J-P)^T$ following linear programming (in order not to overload the notations we consider $c$ as a row vector in this section).  Let us denote by $\A^{>}(\Pi)$ the set of actions $a$ of $\A$ such that ${\bar{c}}^{\Pi} (a) < 0$. We know from linear programming that if ${\bar{c}}^\Pi(a) \geq 0$ for all $a$, then $x^\Pi$ (and thus $\Pi$) is optimal. If $\A^{>}(\Pi)\neq \emptyset$, then we can  swap actions in $\Pi$ with actions in $\A^{>}(\Pi)$ for each state where such an action exists. Let us denote by $\Pi'$ the resulting policy. The proof of the following proposition is given in Appendix \ref{sec:A2}.

\begin{proposition}\label{lem:proper2}
$\Pi'$ is proper and $c \cdot x^{\Pi'} < c \cdot x^\Pi$
\end{proposition}

Because we have a finite number of proper deterministic and stationary policy, we can conclude that Howard's policy iteration algorithm converges in a finite number of steps.

\begin{theorem}
Under Assumption \ref{as:1}, Howard's PI method converges in a finite number of steps.
\end{theorem}

Observe that it is important not to change actions which are not strictly improving. Indeed, in this case it is easy to build deterministic examples where Lemma \ref{lem:proper2} fails (see for instance Fig. \ref{fig:2} in Section \ref{sec:A0}).  As for value iteration, prior to this work policy iteration was not known to converge in this setting. And again, as for VI, unfortunately Howard's Policy Iteration can be exponential in $n$ and $m$  \cite{Fearnley}.

%\begin{lemma} \label{lem:base}
%Let $\Pi$ be a proper, deterministic and stationary policy. Then there exists a base $B^{\Pi}$ of $(P_s)$ such that $x^{\Pi}$ is the basis solution associated with $B^{\Pi}$
%\end{lemma}
%
%\begin{proof}
%Let $B^{\Pi} = \{ j \in \mathcal{A} | \exists i \in \mathcal{S}, \Pi(i) = j\}$. By lemma \ref{lem:sub}, we have $x^{\Pi}(j) > 0$ for all $j \in B^{\Pi}$ and  $x^{\Pi}(j) = 0$ for all $j \notin B^{\Pi}$.
%\end{proof}
%
%\begin{theorem}\label{th:policyiteration}
%Thanks to lemma \ref{lem:base}, we know that we can find a base of $(P_s)$ and a basis solution. Let $\bar{C}_< = \{ j \in \mathcal{A} | \bar{c}_j < 0\}$. Then the simplex algorithm with the pivot rule defined by $\bar{C}_<$ terminates in a finite number of pivots.
%This specific simplex is the Howard Policy Iteration algorithm.
%\end{theorem}

\subsection{The Primal-Dual algorithm~: a generalization of Dijkstra's algorithm}

Primal-dual algorithms proved very powerful in the design of efficient (exact or approximation) algorithms in combinatorial optimization. Edmonds' algorithm for the weighted matching problem \cite{EDM65} is probably the most celebrated example. It is well-known that for the deterministic shortest path problem, when the costs are non negative, the primal-dual approach corresponds to Dijkstra's algorithm \cite{PapSte82}. We extend this approach to the SSP setting. Let us first recall  the linear formulation of the problem and its dual~:

%This is a standard approach for solving linear program that guarantees finiteness of the algorithm when an appropriate rule is used to prevent cycling \cite{PapSte82}.

%We instanciate the primal-dual algorithm for the formulation (P) and its dual (D). 

\setlength{\tabcolsep}{30pt}

\begin{tabular}{lr}
\begin{minipage}{0.3\textwidth}
\begin{equation*}\tag{P}
\begin{array}{ll}
\min & c^Tx \\
(J-P)^T x & =  {\bf 1} \\
x & \geq 0
\end{array}
\end{equation*}

\end{minipage}&
\begin{minipage}{0.3\textwidth}
\begin{equation*}\tag{D}
\begin{array}{ll}
\max & {\bf 1}^T y \\
(J-P) \ y & \leq c  \\
\end{array}
\end{equation*}
\end{minipage}\\

\end{tabular}

The primal-dual algorithm works as follows here. Consider a feasible solution $\bar{y}$ to (D) (initially $\bar{y}=0$ is feasible if $c\geq 0$). Now let $\bar{\A}:=\{a\in \A:{\bf 1}_a^T\ (J-P)\ y = c_a \}$.  We know from complementary slackness that $\bar{y}$ is optimal if and only if there exists $x\geq 0: (J-P)^Tx={\bf 1}$ and $x_a = 0, \forall a\not\in\bar{\A}$ ((P) admits an finite optimum  by Assumption \ref{as:1}). The problem can be rephrased as a so-called restricted primal (RP), where $J_{\bar{\A}}, P_{\bar{\A}}$ and  $x_{\bar{\A}}$ are the restrictions of $J,P,x$ to the row in $\bar{\A}$. We also give its corresponding dual problem (DRP).

\begin{tabular}{lr}
\begin{minipage}{0.4\textwidth}
\begin{equation*}\tag{RP}
\begin{array}{ll}
\min \ \ \ \ \ \ \ {\bf 1}^T z \\
(J_{\bar{\A}}-P_{\bar{\A}})^T x_{\bar{\A}} + z & =  {\bf 1} \\
x_{\bar{\A}}, z & \geq 0\\
\end{array}
\end{equation*}

\end{minipage}&
\begin{minipage}{0.3\textwidth}
\begin{equation*}\tag{DRP}
\begin{array}{ll}
\max & {\bf 1}^T y \\
(J_{\bar{\A}}-P_{\bar{\A}}) \ y & \leq  0 \\
y & \leq \bf 1\\
\end{array}
\end{equation*}
\end{minipage}\\

\end{tabular}

If there is a solution of cost $0$ to (RP) then we have found an optimal solution to our original problem. Else, we use an optimal, positive cost solution $\underline{y}$ to  (DRP) and we update the initial solution by setting $\bar{y}:=\bar{y}+\epsilon \underline{y}$ with $\epsilon\geq 0$ maximum with the property that $\bar{y}+\epsilon \underline{y}$ remains feasible for $(D)$, and we iterate. The algorithm is known to converge in a finite number of steps ((RP) being non degenerate, no anti-cycling rule is needed to guarantee finiteness here \cite{PapSte82}) and this provides an alternative approach to the problem as long as we can also solve (RP) and (DRP).

Observe that (RP) can be interpreted as a SSP problem with action set $\bar{\A}\cup \{m+1,...,m+n\}$, where actions $m+k$, for all $k=1,...,n$ is an artificial action associated with state $k$ that lead to the target state $0$ with probability one. The cost of actions in $\bar{\A}$ is zero while the cost of the artificial actions $m+1,...,m+n$ is one. The primal-dual approach thus reduces the initial problem to a sequence of simpler 0/1 cost SSP problems. Note that (RP) is actually the problem of maximizing the probability of reaching state $0$ using only actions in $\bar{\A}$. This problem is known in the AI community as MAXPROB \cite{Kolobov12}. Little is known about this problem. We know though that it can be solved in weakly polynomial time because it fits into our framework and we can thus solve it using linear programming. We could also use Value Iteration, the simplex method or Policy Iteration as described in the previous subsections. Some simplex rules are known to be exponential in this setting \cite{melekopoglou1994complexity}~:  the question of the existence of a strongly polynomial algorithm is thus wide open for this subproblem too and we believe that MAXPROB deserves attention on its own. Using Howard's policy iteration algorithm to solve the auxiliary problem, the primal-dual approach provides an alternative finite algorithm to solve SSP for non negative costs instances.

\begin{theorem}
When $c\geq 0$, the primal-dual algorithm can be initialized with $\bar{y}=0$ and if the MAXPROB subproblems are solved using Howard's Policy Iteration (or any other simple Policy Iteration method), then it terminates in a finite number of steps.
\end{theorem}

%We would like to point out that in the deterministic setting, when the cost are non negative, it is well-known that the primal-dual approach corresponds to Dijkstra's Algorithm (see again \cite{PapSte82}: in this case  (RP) is a reachability question and it is trivial). Hence our algorithm can be seen as a generalization of Dijkstra's algorithm for the stochastic shortest path problem.  

We are investigating the complexity of this extension of Dijkstra's algorithm to the SSP.  Observe that we do not need to impose that $c$ is non negative to apply the primal-dual approach. In fact, one can use the standard trick of adding an artificial constraint $\sum_a x_a \leq M$ to the problem, with $M$ ``big''  to find an initial dual solution and iterate the algorithm \cite{PapSte82}. The structure of the subproblem changes but it can still be solved using the simplex method. This  provides an alternative approach to Value Iteration and Policy Iteration in the general case too.

{
\bibliographystyle{abbrv} 
\setlength{\itemsep}{0mm}
\setlength{\parsep}{0mm}

%\bibliography{stableset}
}

\appendix
\section{Appendix}
\subsection{An example}\label{sec:A0}

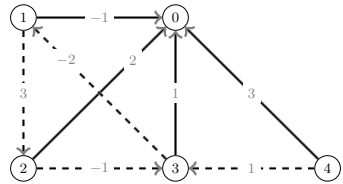
\begin{figure}[H]
\begin{center}
\scalebox{0.8}{
\begin{tikzpicture}
  \begin{scope}[scale=2.5,shape=circle,minimum size=0.5cm,fill=white]
    \tikzstyle{every node}=[draw, scale=0.75] 
    \node (1) at (0,1) {$1$};
    \node (2) at (0,0) {$2$};
    \node (3) at (1,0) {$3$};
    \node (4) at (2,0) {$4$};
    \node (5) at (1,1) {$0$};
    
    %\node (5) at (3.5,1) {$0$} edge [looseness=1,scale=0.5, double=black,in=45,out=-45,loop] node[draw=none,fill=white]  {0} ();
  \end{scope}
  
  \tikzset{mystyle1/.style={->,gray,double=black}} 
  \tikzset{mystyle2/.style={->,dashed,gray,double=black}} 
  \tikzset{every node/.style={fill=white,scale=0.75}}

  \draw (1) edge[mystyle1] node  {$-1$} (5);
  \draw (2) edge[mystyle1] node[near end]  {$2$} (5);
  \draw (3) edge[mystyle1] node  {$1$} (5);
  \draw (4) edge[mystyle1] node  {$3$} (5);
   
   \draw (4) edge[mystyle2] node  {$1$} (3);
   \draw (3) edge[mystyle2] node[near end]  {$-2$} (1);
   \draw (1) edge[mystyle2] node  {$3$} (2);
   \draw (2) edge[mystyle2] node  {$-1$} (3);

\end{tikzpicture}}
\end{center}
\caption{A deterministic shortest path (with target state $0$)~: the dark ``actions'' represent the current policy, and the dashed ``actions'' have non positive reduced cost ; changing all actions with non positive reduced cost yield a new policy which is not proper. }\label{fig:2}
\end{figure}

\subsection{Proof of Lemma \ref{lem:well-def}}\label{sec:A1}
\begin{proof}
$x_{k}^{\Pi}=\Pi^T \cdot P^T\cdot x_{k-1}^{\Pi}$ for all $k\geq1$ and $x_0^\Pi=\Pi^T y_0^\Pi$. Therefore $x_{k}^{\Pi}=(\Pi^T \cdot P^T)^k \cdot \Pi^T y_0^\Pi$, where $y_0^\Pi$ is the original state distribution. It follows that  $\sum_{k=0}^{K} x_{k}^{\Pi}= \sum_{k=0}^{K}((\Pi^T \cdot P^T)^k \cdot \Pi^T y_0) =  (\sum_{k=0}^{K}(\Pi^T \cdot P^T)^k) \cdot \Pi^T y_0$ and because of the standard Lemma \ref{lem:Q}, it implies that $I-\Pi^T \cdot P^T$ is invertible and that $\lim_{K\rightarrow +\infty} \sum_{k=0}^{K}  x_{k}^{\Pi} =  (I-\Pi^T \cdot P^T)^{-1} \cdot \Pi^T y_0$. ($\lim_{k\rightarrow +\infty} \Pi^T \cdot P^T=0$ by definition of BT-properness since ${\bf 1}^T (P^T \cdot \Pi^T)^n \cdot e_i < 1$ for all $i=1,...,n$).\\
\end{proof}

\begin{lemma}\label{lem:Q}
Let $Q$ be a matrix with $\lim_{k\rightarrow +\infty} Q^k=0$. Then $I-Q$ is invertible, $\sum_{k\geq 0} Q^k$ is well defined and $\sum_{k\geq 0} Q^k=(I-Q)^{-1}$.
\end{lemma}

\subsection{Proof of Theorem \ref{th:VI}}\label{sec:A4}

We will prove that $\min_{\Pi \in \mathcal P} \lim_{K \rightarrow +\infty} {\sum_{k = 0}^{K}{c^T x_k^{\Pi}}}=  \lim_{K \rightarrow +\infty} \min_{\Pi \in \mathcal{P}^K} {\sum_{k = 0}^{K}{c^T x_k^{\Pi}}}$ with $y^\Pi_0:=e_i$, for all $i=1,...,n$, by proving both inequalities.

\fbox{\it $\leq$} Let $\Pi^*_K$ be an optimal solution to $\min_{\Pi \in \mathcal{P}^K} {\sum_{k = 0}^{K}{c^T x_k^{\Pi}}}$ computed by dynamic programming (as described above). $\Pi^*_K$ is a proper policy for aux-SSP for all $K$. By feasibility of $\Pi^*_K$, we thus have $V_K(i)=c^T \sum_{k=0}^K x^{\Pi^*_K}_k  \geq  \min_{\Pi \in \mathcal P_{aux}} \lim_{K \rightarrow +\infty} {\sum_{k = 0}^{K}{c^T x_k^{\Pi}}}$ (observe that this minimum is well defined since we are still satisfying Assumptions \ref{as:1} in aux-SSP). By construction $\{V_K(i), K\geq 1\}$ is non-increasing, hence because it is bounded from below, it converges and  $\lim_{K \rightarrow +\infty} V_K(i)$ is well-defined.  Taking the limit we get $\lim_{K \rightarrow +\infty} c^T \sum_{k=0}^K x^{\Pi^*_K}_k  \geq  \min_{\Pi \in \mathcal P_{aux}} \lim_{K \rightarrow +\infty} {\sum_{k = 0}^{K}{c^T x_k^{\Pi}}}$. But $\min_{\Pi \in \mathcal P_{aux}} \lim_{K \rightarrow +\infty} {\sum_{k = 0}^{K}{c^T x_k^{\Pi}}}\geq  \min_{\Pi \in \mathcal P} \lim_{K \rightarrow +\infty} {\sum_{k = 0}^{K}{c^T x_k^{\Pi}}}$ if $M(i)$ is chosen so that auxiliary actions can be assumed not to be used in an optimal policy in $\mathcal P_{aux}$. This is the case for $M(i)\geq V^*(i)$.

\fbox{\it $\geq$} Let $\Pi^*$ be an optimal proper deterministic and stationary solution to $\min_{\Pi \in \mathcal P} \lim_{K \rightarrow +\infty} {\sum_{k = 0}^{K}{c^T x_k^{\Pi}}}$ ($\Pi^*$ exists in our setting by Corollary \ref{cor:detstat}). Let us denote by $\bar\Pi$ the policy of $\mathcal P_{aux}$ that chooses the auxiliary action for each state. Consider the policy $\Pi_K$ of $\mathcal P^K$ obtained from using $\Pi^*$ in periods $0,...,K-1$ and policy $\bar\Pi$ in period $K$. By feasibility of $\Pi_K$, we have $c^T x^{\Pi_K}_K + \sum_{k=0}^{K-1} c^T x_k^{\Pi_K} \geq \min_{\Pi \in \mathcal{P}^K} {\sum_{k = 0}^{K}{c^T x_k^{\Pi}}}$. Now taking the limit as $K$ tends to infinity, we have the result since $\lim_{K\rightarrow +\infty} x^{\Pi_K}_K=\lim_{K\rightarrow +\infty} x^{\Pi^*}_{K}=0$ as $\Pi_K$ differs from $\Pi^*$ only in period $K$, and $\Pi^*$ is $i$-proper.\\

\subsection{Proof of Proposition \ref{lem:proper2}}\label{sec:A2}
\begin{proof}
We denote by $y^\Pi$ the dual solution associated with $\Pi$ i.e. $y^\Pi=c_\Pi (I-P_\Pi)^{-T}$. Assume for contradiction that $\Pi'$ is not proper. Let $G_{\Pi'}$ be the support graph of this policy. Since $\Pi'$ is not proper,  there exists a non empty set of states that are not in $R^-(0)$. It implies that there is a set of vertices $V$ in $G_{\Pi'}$ such that $0_\S, 0_\A \not\in V$ and $\delta^+(V)=\emptyset$. Now there exists an action $a$ of $\A^{>}(\Pi)$ in $V$, otherwise vertices in $V$ are not in $R^-(0)$ in $G_{\Pi}$, a contradiction. Consider the graph $G_a$ obtained by taking the subgraph of $G_{\Pi'}$ induced by the vertices in $V$ that are reachable from $a$, by removing the edge between $a$ and the unique state $s$ with $a\in \A(s)$, and by adding an artificial state $s_0$ with $a$ as its unique possible action. Let $\A_a$ be the set of actions in $G_a$.

We can associate a $s_0$-SSP instance to $G_a$ by considering $s$ as the target state. We can assume w.l.o.g. that $\Pi'$ is a $s_0$-proper policy for this problem\footnote{We can assume without loss of generality that every vertex in $G_a$ is in $R^-(s)$. If not, we change the set $V$ by considering instead the vertices in $G_a$ that do not have a path to $s$ (and we iterate the procedure if $V$ still does not satisfy the required property). Now it is clear that $G_a$ contains at least one state $s'$ and one action $\Pi(s')$. Again not all actions in $G_a$ are actions from $\Pi$ because otherwise $\Pi$ would not be proper.}. Now let $x^{\Pi'}$ be the corresponding flux vector (in principle it is defined only on the actions in $G_a$ but we extend the flux on the other action by setting it to zero). We can interpret  $x^{\Pi'}$ as a (non zero) transition cycle of the original problem  (the flux is defined on the same set of actions and $x^{\Pi'}(a)=1$). The vector ${x}^{\Pi'}\geq 0$ thus satisfies $(J-P)^T{x}^{\Pi'}=0$. Now  the reduced cost ${\bar{c}}^{\Pi}(a')=c(a') - c_\Pi (I-P_\Pi)^{-T} (J-P)^T {\bf 1}_{a'} \leq 0$ for all $a'\in \A_a$ by definition of $\Pi$ and $\Pi'$. Also, as already observed, ${\bar{c}}^{\Pi}(a)<0$. Let us analyze $c {x}^{\Pi'}$. We have $c {x}^{\Pi'}= \sum_{a'\in \A_a} c(a') \cdot {x}^{\Pi'}(a')=(\sum_{a'\in \A_a} {\bar{c}}^{\Pi}(a') \cdot {x}^{\Pi'}(a'))+ c_\Pi (I-P_\Pi)^{-T} (J-P)^T  {x}^{\Pi'}$. Because $(J-P)^T {x}^{\Pi'}=0$, we have $c \  {x}^{\Pi'} = \sum_{a'\in \A_a} {\bar{c}}^{\Pi}(a') \   {x}^{\Pi'}(a')$ but this is negative as $ {x}^{\Pi'}(a')>0$, ${\bar{c}}^{\Pi}(a') \leq 0$ for all $a'\in \A_a$, and ${\bar{c}}^{\Pi}(a) < 0$ .  Therefore ${x}^{\Pi'}$ is a negative cost transition cycle for our original instance, but this contradicts Assumption \ref{as:1}. \\

$c \ x^{\Pi'} - c \  x^\Pi = c \ (x^{\Pi'} - x^\Pi) = ({\bar{c}}^{\Pi} + c_\Pi (I-P_\Pi)^{-T} (J-P)^T) (x^{\Pi'} - x^\Pi)$. But by feasibility $(J-P)^T x^{\Pi'} = (J-P)^T  x^\Pi$ and thus $c \ x^{\Pi'} - c \ x^\Pi = {\bar{c}}^{\Pi}\  (x^{\Pi'} - x^\Pi) = {\bar{c}}^{\Pi} \  x^{\Pi'}$ (as ${\bar{c}}^{\Pi} = 0$ for all $a\in \Pi$ by definition of the current basis) . This latter term is negative as $\Pi'$ is using at least one action in $\A^{>}(\Pi)$ and the actions in $\Pi$ have reduced cost zero.
\end{proof}

\end{document}